\documentclass{article}

\usepackage{graphicx,times,amsmath, amssymb}
\usepackage[T1]{fontenc} 
\usepackage{fixltx2e} 
\usepackage{multirow}
\usepackage[table]{xcolor}
\usepackage{parskip}
\usepackage{authblk}

\usepackage[ruled,vlined, linesnumbered, commentsnumbered]{algorithm2e}

\SetAlFnt{\small}
\SetAlCapFnt{\small}
\SetAlCapNameFnt{\small}
\SetAlCapHSkip{0pt}
\IncMargin{-\parindent}

\usepackage{xspace}
\usepackage{url}

\usepackage{xcolor}

\newtheorem{theorem}{Theorem}[section]
\newtheorem{lemma}[theorem]{Lemma}

\newenvironment{proof}[1][Proof]{\begin{trivlist}
\item[\hskip \labelsep {\bfseries #1}]}{\end{trivlist}}

\begin{document}

\title{Distributed Monitoring for Prevention of Cascading Failures in  Operational Power Grids}

\author[1]{Martijn Warnier\thanks{Corresponding Author.
    email:M.E.Warnier@tudelft.nl, tel:+31 (0)15 27 82232}}
\author[2]{Stefan Dulman}
\author[1]{Yakup Ko\c{c}}
\author[2]{Eric Pauwels}
\affil[1]{Systems Engineering, Faculty of Technology, Policy and Management,
Delft University of Technology, Delft}
\affil[2]{Intelligent Systems group, Centrum Wiskunde \& Informatica
  (CWI), Amsterdam}

\maketitle


\begin{abstract}
  Electrical power grids are vulnerable to cascading failures that can
  lead to large blackouts. Detection and prevention of cascading
  failures in power grids is important. Currently, grid operators
  mainly monitor the state (loading level) of individual components in
  power grids. The complex architecture of power grids, with many
  interdependencies, makes it difficult to aggregate data provided by
  local components in a timely manner and meaningful way: monitoring
  the resilience with respect to cascading failures of an operational
  power grid is a challenge.

  This paper addresses this challenge. The main ideas behind the paper
  are that (i) a robustness metric based on both the topology and the
  operative state of the power grid can be used to quantify power grid
  robustness and (ii) a new proposed a distributed computation method
  with self-stabilizing properties can be used to achieving near
  real-time monitoring of the robustness of the power grid. Our
  contributions thus provide insight into the resilience with respect
  to cascading failures of a dynamic operational power grid at
  runtime, in a scalable and robust way. Computations are pushed into
  the network, making the results available at each node, allowing
  automated distributed control mechanisms to be implemented on top.
\end{abstract}

\section{Introduction}
\label{sec:introduction}

Power grids represent critical infrastructure: all kind of services
(basic services, governmental and private) depend on the continuous
and reliable delivery of electricity. Power grid outages have a large
effect on society, both in terms of safety and in terms of economic
loss. The large-scale introduction of ``renewable energy sources'' and
the current (centralized) architecture of the power grid make it more
likely that large power outages will become more common. Encouraged by
government subsidies and a trend to become more ``green'', consumers
are becoming producers of electricity by installing solar panels and
wind mills~\cite{masters2013renewable}. Part of this produced power
will be used locally, but excess power can be sold and fed back into
the power grid. This in turn leads to grid
instability~\cite{zhang2010impact}: it is more difficult to predict, 
and hence balance, electricity production when there is a large
amount of small producers spread over a large geographical region,
instead of a couple of large producers. The current power grid
architecture does not support the introduction of renewables at this
scale~\cite{Albert2004}. 

The current organization of the power grid thus makes larger grid
failures more likely to occur: initial local disruptions can spread to
the rest of a power grid evolving into a system-wide outage. In a
power grid, an initial failure can, for example, be caused by an
external event such as a storm, and spreads to the rest of the network
in different ways including due to causes such as instability of
voltage and frequency, hidden failures of protection systems, software
or operator errors, and line overloads. For example, in the case of
cascades due to line overloads, an overloaded line is "tripped'' by a
circuit breaker. At this point electricity can no longer flow through
the line, and the power contained in the line flows to other
lines. This might lead to overloading (part of) these lines causing
them to be tripped as well. As this process repeats over and over
again, more lines are shut down, leading to a \emph{cascading failure}
of the power grid~\cite{Dobson2002,Vailman2011}. Cascading effect due
to line overloads, and preventing such cascading failures form the main
focus of this paper.

In order to detect (and ultimately prevent) cascading failures it is
necessary to monitor (and alter) the current state (power load
distribution) of the power grid. The emerging Smart Grid provides
exactly this: a power grid with a communication overlay that connects
sensors and effectors. In effect, a Smart Grid is a large-scale
distributed system that enables the monitoring of line loads and that
enables changing the state of the network by tripping and untripping
lines. In the remainder of the paper a Smart Grid is assumed.

Given this context of the Smart Grid, this paper addresses two main
research questions: \emph{What should be monitored?}, i.e., is there a
\emph{metric} that can be used for cascading failure prediction?
\emph{How to monitor?}, i.e., how should aggregation be performed and
which temporal resolution is required for the monitoring. In addition,
it should be possible to extend the proposed (passive) monitoring
scheme to an (active) scheme that automatically alters the state of
the grid in order to prevent cascading failures. 

The main contribution of the paper is a new distributed monitoring
approach that can be used to monitor the robustness of the power grid
with respect to cascading failures. The monitoring approach is
based on the distributed computation of the robustness metric we
introduced in~\cite{Koc:373,Koc:862}. Our contributions in this paper
include the extension of a distributed gossiping
algorithm~\cite{boyd2005gossip} with self-stabilization mechanisms to
account for network dynamics. The resulting framework allows
distributed aggregates to be computed fast and reliable, which forms
the core of the proposed monitoring approach.

Our main results show that we are able to compute the complex
robustness metric using simple robust distributed primitives with
results readily made available at each node in the network. This is an
important property as the mechanisms presented in this paper can be
seen as a measurement framework to be used in real-time for the design
of distributed control mechanisms. Our approach scales very well with
network size (logarithmic order) in terms of convergence time. The
precision of the computations can be fixed by changing the message
sizes and is independent on the network parameters (number of nodes, diameter, etc.).

The remainder of this paper is organized as follows: Section~\ref{sec:robustn-metr-monit} introduces the metric
used to assess the robustness of the power grid with respect to
cascading failures.  Section~\ref{sec:decentr-aggr} presents the
distributed algorithm for the online computation of the robustness
metric. Section~\ref{sec:analysis-discussion} discusses the simulation
results that show the applicability the proposed
approach. Section~\ref{sec:syst-health-monit} presents the 
current state of the art in power grid monitoring and cascading failure 
detection. Section~\ref{sec:conclusions} concludes the paper.

\section{Robustness Metric}
\label{sec:robustn-metr-monit}

Different topological metrics have been identified in literature that
indicate the vulnerability of a power grid against cascading failures
on the basis of which the most critical nodes in a network are
identified. Examples of such topological metrics are average shortest
path length, betweenness centrality~\cite{freeman1977set} and the gap
metric~\cite{el1985gap}.  However, next to a topological aspect, power
grids also have a physical aspect. In particular, electrical current
in a power grid behaves according to Kirchoff's
laws~\cite{belevitch1962summary}. A metric that quantifies the
robustness of an \emph{operational} power grid with respect to
cascading failures should take both these aspects into account. Our
robustness metric from~\cite{Koc:373,Koc:862} does exactly this, and
it therefore forms the starting point for the distributed power grid
monitoring algorithm proposed in this paper. The robustness metric
$R_{CF}$ (for Robust against Cascading Failures) assess the robustness
of a given power grid with respect to cascading failures due to line
overloads. The metric relies on two main concepts: electrical nodal
robustness and electrical node significance. Higher values of $R_{CF}$
indicate a robuster, i.e., more able to resist cascading failures,
power grid. The remainder of this section provides a summary from our
earlier work on robustness metrics, we refer to~\cite{Koc:373,Koc:862}
for more details.

\subsection{Electrical Nodal Robustness}
\label{subsec_electrical nodal robustness}

The \emph{electrical nodal robustness} quantifies the ability of a bus
(i.e. a node in a graph representation of a power grid) to resist the
cascade of line overload failures by incorporating both flow dynamics
and network topology. In order to calculate this value for a node,
three factors are of importance: (i) the homogeneity of the load
distribution on out-going branches (i.e. links in a graph
representation of a power grid); (ii) the loading level of the
out-going links; and (iii) the out-degree of the node.

Entropy is used to capture the first and the last factors described
above: the entropy of a load distribution at a node increases as flows
over lines are distributed more homogeneously and the node out-degree
increases. The entropy of a given load distribution at a node $i$ is
computed by Equation~\eqref{Entropy}:

\begin{equation}\label{Entropy}
H_{i}=\sum_{j=1}^{d} p_{ij} \log p_{ij}
\end{equation}

 where $d$ refers to the out-degree of the corresponding node, whereas $p_{ij}$ corresponds to normalized flow values on the out-going links \emph{$l_{ij}$}, given as:

\begin{equation}\label{eq:normflow}
p_{ij}=\frac{f_{ij}}{\sum_{j=1}^{d} f_{ij}}
\end{equation}

where $f_{ij}$ refers to the flow value in line \emph{$l_{ij}$}. To
model the effect of the loading level of the power grid the tolerance
parameter $\alpha$ is used (see~\cite{Motter2002}). The tolerance
level of a line \emph{$l_{ij}$}, $\alpha_{ij}$, is the ratio between
the rated limit and the load of the corresponding line
\emph{$l_{ij}$}.

Combining Equations~\eqref{Entropy} and \eqref{eq:normflow} with the
tolerance parameter $\alpha$ to capture the impact of loading level on
the robustness, the electrical nodal robustness of a node $i$
(i.e. $R_{n,i}$), which takes both the flow dynamics and topology
effects on network robustness into account, is then defined as:

\begin{equation}\label{eq:rn}
R_{n,i}=-\sum_{j=1}^{d} \alpha _{ij}p_{ij} \log p_{ij}
\end{equation}

In Equation~\eqref{eq:rn}, the minus sign (-) is used to compensate
the negative electrical nodal robustness value that occurs due to
taking the logarithm of normalized flow values.

\subsection{Electrical Node Significance}
\label{subsec_Electrical node significance}

Not all nodes in a power grid have the same influence on the
occurrence of cascading failures. Some nodes distribute a relatively
large amount of the power in the network, while other nodes only
distribute a small amount of power. When a node (or line to a node)
that distributes a relatively large amount of power fails, the result
is more likely to lead to a cascading failure, ultimately resulting in
a large grid blackout. In contrast, if a node that only distributes a
small amount of power fails, the resulting redistribution of power can
usually be accommodated by the other parts of the network.  Thus, node
failures have a different impact on the context of cascading failure
robustness and this impact depends on the amount of power, distributed
by the corresponding node. The impact of a particular node is
reflected by the electrical node significance $\delta$, which is:
%
\begin{equation}\label{delta}
\delta _{i}=\frac{P_{i}}{\sum_{j=1}^{N} P_{j}},
\end{equation}
where $P_{i}$ stands for total power distributed by node $i$ while, $N$
refers to number of nodes in the network.
Electrical node significance is a centrality measure that can be used
to rank the relative importance (i.e., criticality) of nodes in a power
grid in the context of cascading failures. Failures of nodes with
a higher $\delta$ will typically result in larger cascading failures.

\subsection{Network Robustness Metric}
\label{subsubsec_Network robustness metric}

The network robustness metric $R_{CF}$~(\cite{Koc:373,Koc:862}) is obtained 
by combining the nodal robustness and node significance:
%
\begin{equation}\label{rcf}
R_{CF}=\sum_{i=1}^{N} R_{n,i}\delta _{i}.
\end{equation}
The above metric can be used as a robustness indicator for
power grids. This is done as follows: for a normally operating power
grid the robustness metric is calculated, which results in some value
$v$. This value is used as a base case. During normal operation the
robustness metric value will change somewhat, because different nodes
will demand different electricity quantities over time, leading to
different loading levels in the network. However, a larger change in
the robustness metric, a drop in particular, indicates that a
cascading failure becomes more likely and grid operators may need to
take evasive actions (e.g., adding reserve capacity to the grid or
demand shifting of power). Note that, in the general case, it is
complicated to determine what good safety margins are, or for which
values of the robustness metric the exact tipping point is located
(i.e., the point where a small failure will lead to a massive
blackout). Ultimately this needs to be determined by the grid
operators. We have determined this point experimentally, by
simulation, for a specific power grid: the IEEE 118 Power system (see
Section~\ref{subset:robust-metr-comp} ). We refer to~\cite{Koc:571}
which presents a more general and structured investigation of this topic.

\section{Decentralized Aggregation}
\label{sec:decentr-aggr}

The computation of the robustness metric introduced in the previous section in a centralized manner raises a number of challenges when applied to large areas (i.e., provinces or even whole countries). Scalability, single-point-of-failure, real-time results dissemination, fault tolerance, maintenance of dedicated hardware are just a few examples that hint towards a decentralized approach as a more convenient solution.

The described problem maps onto a geometric random graph (mesh network), where the nodes can communicate mainly with their direct neighbors.  
From the perspective of the communication model, we assume that time is discrete. During one \emph{time step} each node will pick and communicate with a random neighbor. Major updates in the network occur just once in a while (for example, in the described scenario, new measurement data is made available once every 15 minutes). We will make use of the concept of \emph{time rounds} and ask the nodes to update their local data at the beginning of the rounds. The bootstrap problem and round-based time models received a lot of attention in literature~\cite{jelasity2005gossip, bicocchi2010handling, pruteanu2012lossestimate} - in our application scenario the constraints being very loose allow for an algorithm like the one presented in~\cite{werner2005firefly}.

We make no assumptions with respect to nodes stop-failing or new nodes joining the network. The mechanism described below can accommodate these cases and the computation results will adapt themselves to such changes.

\subsection{Solution Outline}

Our solution for computing the robustness metric uses a primitive for computing sums in a distributed network inspired by the gossip-alike mechanism presented in~\cite{mosk2008fast} (see Figure~\ref{fig:example}). The algorithm presented in~\cite{mosk2008fast} computes a sum of values distributed on the nodes of a network by using a property of order statistics applied to a series of exponential random variables. The algorithm resembles gossiping algorithms~\cite{jelasity2005gossip} but differs in a number of important points. 

Essentially, it trades communication for convergence speed. By relying on the propagation of an extreme value (the minimum value in this case), locally computable, it achieves the fastest possible convergence in a distributed network - $O(D \log N)$ time steps ($D$ is the diameter of the network and $N$ the number of nodes). This speed is significant compared to the original gossiping algorithms that converged in $O(D^2 \log N)$ time steps~\cite{boyd2005gossip}. For example,  in Figure~\ref{fig:example} a $N=1000$ nodes network with diameter $14$ converges after the first $15$ computation steps. 
The paid price is the increased messages size $O(\delta^{-2})$, where $\delta$ is a parameter defining the precision of the final result. Assuming $s$ as the ground-truth result, the algorithm offers an estimate in the interval $[(1-\delta) s, (1+\delta) s]$ with an error $\epsilon = O(1/poly(N))$.

\begin{figure}
  \centering
  \includegraphics[width=0.95\columnwidth]{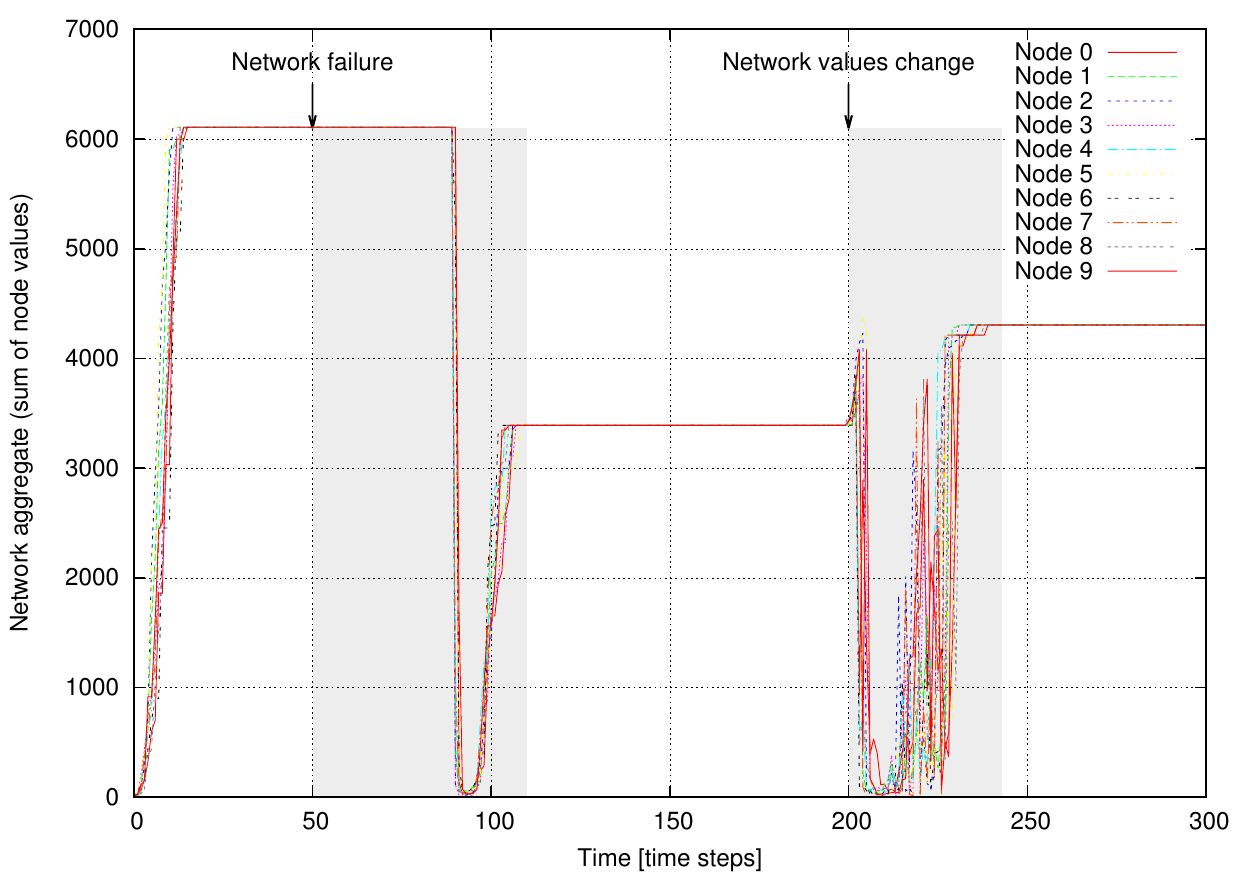}
  \caption{Sum computation during network dynamics (geometric random graph with $1000$ nodes initially, diameter 14, random values, half of the network is disconnected at time 50, nodes change their values at at time 200).}
  \label{fig:example}
\end{figure}

We extend the extreme value propagation mechanisms to account for dynamics in the network. Specifically, we add a \emph{time-to-live field} to each value - an integer value that decreases with time and marks the age of the current value. This mechanism takes care of nodes leaving the network, stop-crashing or resetting. In the example in Figure~\ref{fig:example}, after convergence, we removed half of the nodes in the network at time $50$. The effect of expiring time-to-live (set to a maximum of $50$ in this example) can be seen around the time step $100$.
Furthermore, we extend the time-to-live expiry mechanism to achieve a $O(D \log N + \log T)$ time steps \emph{value removal}. In other words, if a certain extreme value propagated through the network, we mark it as ``expired'' and assure its associated time-to-live value to expire (reach 0) within $O(D \log N + \log T)$ time steps. This is shown in Figure~\ref{fig:example} in the interval $200-300$. At time $200$ half of the nodes in the network changed their values randomly triggering the expiration mechanism.

Our distributed approach solves most of the scaling issues and proves to be highly robust against network dynamics (e.g., network nodes becoming unavailable due to failures, reconfiguration, new nodes joining the system, etc.). As we show in the following, our approach is very fast for a typical network, outperforming by far the speed of centralized approaches. As the protocols rely on anonymous data exchanges, privacy issues~\cite{IEEEPS09} are alleviated, as the identities of the system participants are not needed in the computations.

The downsides of our approach map onto the known properties of this class of epidemic algorithms. Although anonymity is preserved, an authentication system~\cite{jesi2007identifying} is needed to prevent malicious data corrupting the computations. Also, a light form of synchronization~\cite{werner2005firefly} is needed for coordinating nodes to report major changes in their local values - fortunately, the nature of the problem we address here allows it.

\subsection{Self-stabilizing Sum Computation - $ComputeSum()$}

The basic mechanism behind the sum computation algorithm presented below relies on minimum value propagation via gossiping. Assume that each node holds a positive value $x_i$. At each time step, each node chooses a random neighbor and they exchange their values, both keeping the smallest value. The smallest value propagates fast in the network, in $O(D \log N)$ time steps, via this push-pull gossiping mechanism (see \cite{shah2009gossip} Section 3.2.2.4 page 32).

\IncMargin{1em}
\begin{algorithm}[t!]
\DontPrintSemicolon
\SetAlgoNoEnd
\SetNoFillComment
    \BlankLine
    \tcc*[l]{\emph{$v$, $\tau$ - received value and time-to-live}}
        \tcc*[l]{\emph{$v_{local}$, $\tau_{local}$ - local value and time-to-live}}
    \BlankLine
    \tcc*[l]{\emph{create temporary variables}}
    $(v_m, v_M) \leftarrow \left( \min(v, v_{local}), \max(v, v_{local}) \right)$ \label{a1:l0} \;
    $(\tau_m, \tau_M) \leftarrow$ corresponding $(\tau, \tau_{local})$ to $(v_m, v_M)$ \;
    \BlankLine    
    \tcc*[l]{\emph{update logic}}
    \If{$v_m==v_M$}{
      \If(\tcc*[f]{\emph{equal negative values}}){$v_m < 0$}{
        $\tau_m \leftarrow \mathcal{C} \tau_m$ \label{a1:l3}}
      \Else(\tcc*[f]{\emph{equal positive values}}){
        $\min(\tau_m, \tau_M) \leftarrow \max(\tau_m, \tau_M) - 1$ \label{a1:l2}
      }
    }
    \Else{
      \If(\tcc*[f]{\emph{at least one negative value}}) {$v_m < 0$}{
        \If {$v_m==-v_M$}{
          $(\tau_m, \tau_M) \leftarrow (T, T)$}
        \Else {
          $(\tau_m, \tau_M) \leftarrow (\mathcal{C} \tau_m, \mathcal{C} \tau_M)$ \label{a1:l4}}      
      }
      \Else(\tcc*[f]{\emph{two different positive values}}) {
        $\tau_M \leftarrow \tau_m - 1$
      }
    }
    \BlankLine    
    \tcc*[l]{\emph{update local variables}} 
    $(v, v_{local}) \leftarrow$ $(v_m, v_m)$ \label{a1:l1} \;
    $(\tau, \tau_{local}) \leftarrow$ corresponding $(\tau_m, \tau_M)$ \;
\caption{PropagateMinVal($v, \tau$)}
\label{alg:ttl}
\end{algorithm}
\DecMargin{1em}

Assume that each node $i$ in the network holds a positive value $x_i$. In order to compute the sum of all $n$ values in the network $\left(\sum_{i=1}^N x_i\right)$, the authors of~\cite{mosk2008fast} propose that each node holds a vector $\mathbf{v}$ of $m$ values, initially drawn from a random exponential random distribution with parameter $\lambda_i=x_i$. After a gossiping step between two nodes $i$ and $j$, the vectors $\mathbf{v}_i$ and $\mathbf{v}_j$ become equal and hold the minimum value on each position of the initial vectors. Thus, given an index $k \in (1,m)$, the resulting vectors $\mathbf{v}'_i$, $\mathbf{v}'_j$ will have the property $\mathbf{v}'_i[k] = \mathbf{v}'_j[k] = \min\left(\mathbf{v}_i[k], \mathbf{v}_j[k]\right)$. The authors show that, after all vectors converge to some value $\mathbf{v}$, the sum of $x_i$ values in the network may be approximated by: $\sum_{i=1}^N x_i = \frac{m}{\sum_{k=1}^m \mathbf{v}[k]}$ (see~\cite{shah2009gossip} Section 5.2.5.4 page 75).

We extend the algorithm presented in~\cite{mosk2008fast} by adding to each node a new vector $\boldsymbol{\tau}_i$ holding a time-to-live counter for each value. This new vector is initialized with a default value $T$, larger than the convergence time of the original algorithm (choosing a proper value is explained below). The values in $\boldsymbol{\tau}_i$ decrease with $1$ every time slot, with one exception. The node generating the minimum $\mathbf{v}_i[k]$ on the position $k \in (1, m)$ sets $\boldsymbol{\tau}_i[k]$ to $T$ (see Algorithm~\ref{alg:shaplus} line~\ref{a2:l0}). In the absence of any other dynamics, all properties proved in~\cite{shah2009gossip} remain unchanged as the output of our approach is identical to the original algorithm.

The main reason for adding the time-to-live field is to account for nodes leaving the network or nodes that fail-stop. We avoid this way complicated mechanisms in which nodes need to keep track of neighbors. Additionally, this mechanism does not make use of node identifiers. 
The intuition behind this mechanism is that a node generating the network-wide minimum on position $k \in(1,m)$ will always advertise it with the accompanying time-to-live set to the maximum $T$. The rest of the nodes will adopt the value $\mathbf{v}[k]$ and have a value $\boldsymbol{\tau}[k]$ decreasing with the distance from the original node. $T$ is chosen to be larger than the maximum number of gossiping steps it takes the minimum to reach any node in the network. In a gossiping step between two nodes $i$ and $j$, if $\mathbf{v}_i[k]=\mathbf{v}_j[k]$ then the largest of the $\boldsymbol{\tau}_i[k]$ and $\boldsymbol{\tau}_j[k]$ will propagate (Algorithm~\ref{alg:ttl} line~\ref{a1:l2}). This means that $\boldsymbol{\tau}[k]$ on all nodes will be strictly positive for as long as the node is online. If the node that generated the minimum value on the position $k$ goes offline, all the associated $\boldsymbol{\tau}[k]$ values in the network will steadily decrease (Algorithm~\ref{alg:shaplus} line~\ref{a2:l1}) until they will reach $0$ and the minimum will be replaced by next smallest value in the network (Algorithm~\ref{alg:shaplus} lines~\ref{a2:l2}-\ref{a2:l3}). It will take $T$ time steps for the network to ``forget'' the value on position $k$. The graphical effect of this $O(T)$ mechanism is shown in Figure~\ref{fig:example} in the interval $50-150$.

\IncMargin{1em}
\begin{algorithm}[t!]
\DontPrintSemicolon
\SetAlgoNoEnd
\SetNoFillComment
  \BlankLine
  \tcc*[l]{\emph{$\mathbf{v^0}$ - original random samples vector on this node}}  
  \tcc*[l]{\emph{$\mathbf{v}, \boldsymbol{\tau}$ - received value and time-to-live vectors}}
  \BlankLine
  \tcc*[l]{\emph{update all elements in the data vector}}
  \For{$j=1$ to length$(\mathbf{v})$}{ 
    PropagateMinVal$(\mathbf{v}[j], \boldsymbol{\tau}[j])$
  }
  \BlankLine
  \tcc*[l]{\emph{time-to-live update - do once every timeslot}}
  \For{$j=1$ to length$(\mathbf{v})$}{    
    \If(\tcc*[f]{\emph{reinforce a minimum}}){$\mathbf{v}[j]==\mathbf{v^0}[j]$} {
      $\boldsymbol{\tau}[j] \leftarrow T$ \label{a2:l0}
    }
    \Else{ 
      $\boldsymbol{\tau}[j] \leftarrow \boldsymbol{\tau}[j] - 1$ \label{a2:l1}
      \tcc*[r]{\emph{decrease time-to-live}}
      \If(\tcc*[f]{\emph{value expired}}){$\boldsymbol{\tau}[j]<=0$}{ \label{a2:l2}
        $\mathbf{v}[j] \leftarrow \mathbf{v^0}[j]$ \;     
        $\boldsymbol{\tau}[j] \leftarrow T$ \label{a2:l3}
      }
    }
  }
  \BlankLine
  \tcc*[l]{\emph{estimate the sum of elements}} \label{a2:l5}
  $s \leftarrow 0$ \; 
  \For{$j=1$ to length$(\mathbf{v})$}{
    $s \leftarrow s + abs(\mathbf{v}[j])$ \label{a2:l4}
  }
  \BlankLine
  \Return{length$(\mathbf{v}) / s$} \label{a2:l6}
\caption{ComputeSum $(\mathbf{v}, \boldsymbol{\tau})$}
\label{alg:shaplus}
\end{algorithm}
\DecMargin{1em}

The second self-stabilizing mechanism targets nodes changing their values at runtime. Assume a node changes its value $x_i$ to $x'_i$ at some time $t$. This change will trigger a regeneration of its original samples from the exponential random variable $\mathbf{v}_i$ to $\mathbf{v}'_i$. Let $k$ be an index with $k \in (1, m)$. Let $\mathbf{u}$ be the vector containing the minimum values in the network if the node $i$ would not exist. In order to understand the change happening when transitioning from $x_i$ to $x'_i$ we need to look at the relationship between the individual values $\mathbf{v}_i[k]$, $\mathbf{v}'_i[k]$ and $\mathbf{u}[k]$. As shown in Table~\ref{tab:vals}, if $\mathbf{u}[k]$ is the smallest of all three values then no change will propagate in the network. If $\mathbf{v}'_i[k]$ is the smallest value, then this will propagate fast, in $O(D \log N)$ time steps, with the basic extreme propagation mechanism. If $\mathbf{v}'_i[k]$ is the smallest then this value will remain in the network until its associated time-to-live field will expire. As usually $T \gg D$ we add a mechanism to speed up the removal of this value from the network.

The removal mechanism is triggered by the node owning the value that needs to be removed (in our case node $i$) and works as follows: node $i$ will mark the value $\mathbf{v}_i[k]$ as ``expired'' by propagating a negative value $-\mathbf{v}_i[k]$. This change will not affect the extreme value propagation mechanism (see Algorithm~\ref{alg:ttl}) nor the estimation of the sum (notice the use of the absolute value function in Algorithm~\ref{alg:shaplus} line~\ref{a2:l4}). If node $i$ contacts a node also holding the value $\mathbf{v}_i[k]$ then first, it will propagate the negative sign for the value, also maximizing its time-to-live field to a large value $T$. Intuitively, as long as the $\mathbf{v}_i[k]$ is present in the network, the $-\mathbf{v}_i[k]$ will propagate, over-writing it. Considering the large range of unique float or double numbers versus the number of values in a network at a given time, we can safely assume the values in the network to be unique.

The time-to-live field of any negative value will halve with each gossiping step (for $\mathcal{C}=2$) if it does not meet the $\mathbf{v}_i[k]$ value (Algorithm~\ref{alg:ttl} lines~\ref{a1:l3},~\ref{a1:l4}). Intuitively, if a negative value is surrounded by values other than $\mathbf{v}_i[k]$, it will propagate while canceling itself at the same time with an exponential rate. This mechanism resembles somewhat a predator-prey model~\cite{Arditi1989311}, where prey is represented by the $\mathbf{v}_i[k]$ variable and predators by $-\mathbf{v}_i[k]$. We designed it such that the populations cancel each-others, targeting the fixed point at the origin as the solution for the accompanying Lotka-Volterra equations.

\begin{table}[t!]
\centering
\small
\begin{tabular}{ |c|c|c|c|c| }
\hline
\rowcolor{lightgray} Propagation & Ordering & Previous & Intermediate & Final \\
\hline
\multirow{2}{1.5em}{none} & $ \mathbf{u}[k]  < \mathbf{v}_i[k]  < \mathbf{v}'_i[k] $ & $\mathbf{u}[k]$ & $\mathbf{u}[k]$  & $\mathbf{u}[k]$  \\
                          & $ \mathbf{u}[k]  < \mathbf{v}'_i[k] < \mathbf{v}_i[k]  $ & $\mathbf{u}[k]$ & $\mathbf{u}[k]$  & $\mathbf{u}[k]$  \\
\hline
\multirow{2}{1.5em}{slow} & $ \mathbf{v}_i[k]  < \mathbf{u}[k]  < \mathbf{v}'_i[k] $ & $\mathbf{v}_i[k]$ & $\mathbf{v}_i[k]$  & $\mathbf{u}[k]$  \\
                          & $ \mathbf{v}_i[k]  < \mathbf{v}'_i[k] < \mathbf{u}[k]  $ & $\mathbf{v}_i[k]$ & $\mathbf{v}_i[k]$  & $\mathbf{v}'_i[k]$ \\
\hline
\multirow{2}{1.5em}{fast} & $ \mathbf{v}'_i[k] < \mathbf{u}[k]  < \mathbf{v}_i[k]  $ & $\mathbf{u}[k]$ & $\mathbf{v}'_i[k]$ & $\mathbf{v}'_i[k]$ \\
                          & $ \mathbf{v}'_i[k] < \mathbf{v}_i[k]  < \mathbf{u}[k]  $ & $\mathbf{v}_i[k]$ & $\mathbf{v}'_i[k]$ & $\mathbf{v}'_i[k]$ \\
\hline
\end{tabular}
\\ $\,$
\caption{Value propagation.}
\label{tab:vals}
\end{table}

\begin{lemma}
\label{lemma:removal}
  \emph{Value removal delay}

  \parindent = 0pt
  \emph{By using the value removal algorithm, the new minimum propagates in the network in $O(D \log N +\log T)$ time steps.}
\end{lemma}
\begin{proof}
In the worst case scenario, the whole network contains the minimum value $\mathbf{v}_i[k]$ on position $k$, with the time-to-live field setup at maximum $T$.The negative value, being the smallest one in the network, propagates in $O(D \log N)$ in the whole network. Again, in the worst case scenario, we will have a network with each node having the value $-\mathbf{v}_i[k]$ on position $k$ with the time-to-live set to the maximum $T$. From this moment on, the time-to-live will halve at each gossip step on each node (for $\mathcal{C}=2$), reaching $0$, in the worst case scenario in $O(\log T)$ time steps. This is the worst case because nodes may be contacted by several neighbors during a time step leading to a much faster cancellation. Overall, the removal mechanism will be active for at most $O(D \log N +\log T)$ time steps. 
This bound is an upper bound. In reality the spread and cancellation mechanisms will act in parallel, leading to tighter bounds. 
\end{proof}

This result gives us the basis for choosing the $T$ constant. Ideally, it should be chosen as small as possible, in line with the diameter of the network. The fact that the removal mechanism is affected only by $\log T$ lets us use an overestimate of $T$, which can be a few orders of magnitude larger than the diameter of the network, with little impact on the convergence speed. For example, if the network diameter is between $10-30$ and the values refresh each 10000 time steps, we can safely set $T$ anywhere between $1000-10000$ (see Section~\ref{subset:scalability}). This will not affect the convergence of the sum computation mechanism but allow for a timely account for a node removal.

All the mechanisms presented in this section lead to the sum computation mechanism $ComputeSum()$ presented in Algorithm~\ref{alg:shaplus}. It holds the properties of the original algorithm described in~\cite{mosk2008fast} and it additionally showcases self-stabilization properties to account for network dynamics in the form of node removal and nodes changing their values in batches.

\begin{figure}
  \centering
  \includegraphics[width=0.95\columnwidth]{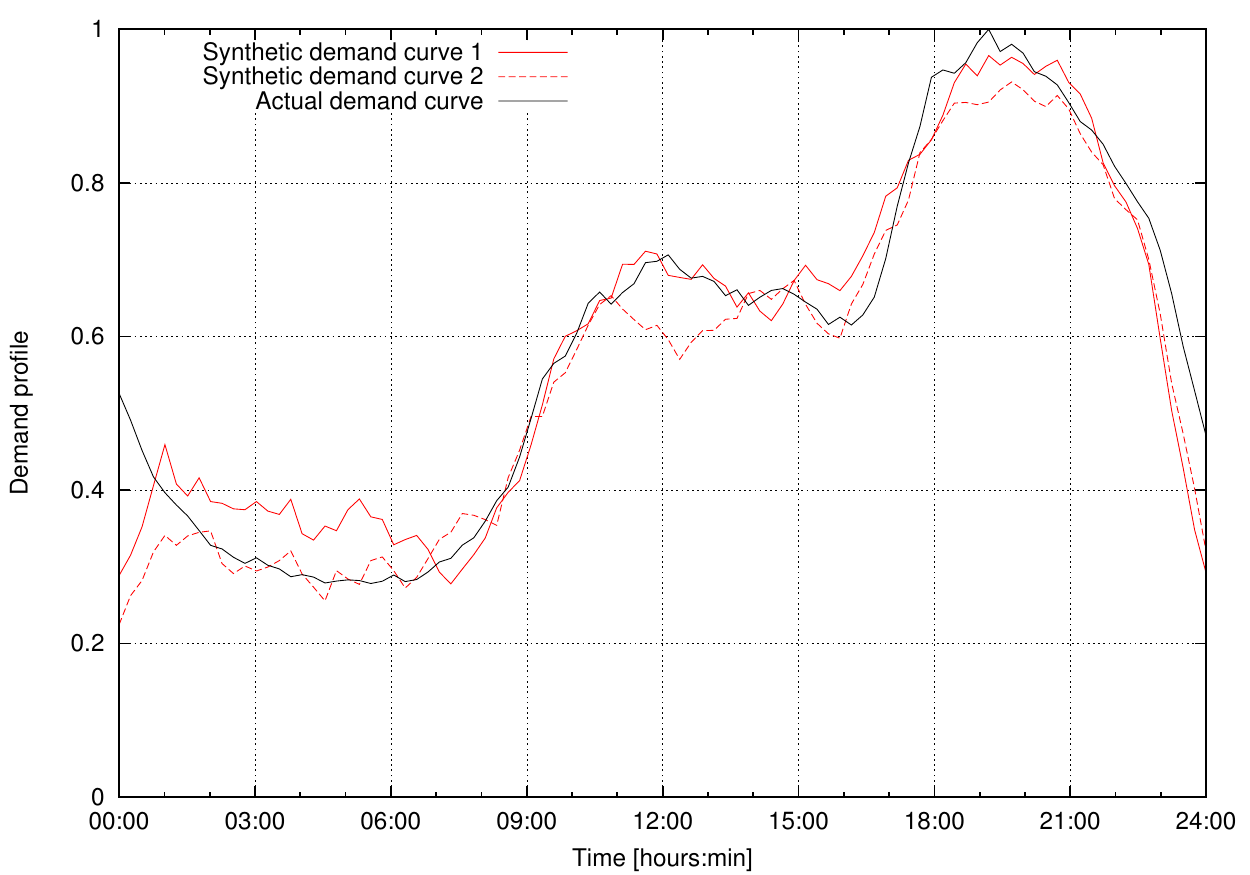}
  \caption{The actual demand profile from a point in Dutch transmission grid and two synthetically generated demand profiles.}
  \label{fig:loadingProfileDutchGrid}
\end{figure}

\subsection{Robustness Metric Computation}

The robustness metric (see Section~\ref{sec:robustn-metr-monit}) is made up of two terms that can be computed locally ($p_i$ in Equation~\eqref{eq:normflow} and $R_{n,i}$ in Equation~\eqref{eq:rn}) and two that can be computed in a distributed fashion ($\delta_i$ in Equation~\eqref{delta} and $R_{CF}$ in Equation~\eqref{rcf}). Equation~\eqref{rcf} can be rewritten as:
\begin{align}
  R_{CF} &= \frac{\sum_{i=1}^N R_{n,i}P_i}{\sum_{j=1}^N P_j},
\end{align}
leading to a solution with two $ComputeSum()$ algorithms in parallel. The first algorithm will compute $\sum_{i=1}^N R_{n,i}P_i$, while the second one will compute $\sum_{j=1}^N P_j$.

Characterizing the convergence time of a composition of distributed algorithms is a difficult task in general. Fortunately, in our case, the composition of the two $ComputeSum()$ has the convergence time equal to each of the two mechanisms, leading to the same $O(D \log N + \log T)$ time steps complexity. Assume the network is stabilized - once the power distributions $P_i$ change both the values $\sum_{i=1}^N R_{n,i}P_i$  and $\sum_{j=1}^N P_j$ will stabilize in $O(D \log N + \log T)$ \emph{in parallel}, as they do not require intermediate results from each other. 

As the type of gossiping algorithms we use are based on minimum value propagation, all the nodes in the network will have the same value once the algorithm converged. Stabilization can be easily detected locally by monitoring the lack of changes in the propagated values for a fixed time threshold.

\section{Analysis and Discussion}
\label{sec:analysis-discussion}

Our approach of computing the robustness metric is scalable and robust. In this section we will focus on some of quantitative aspects, analyzing results obtained from simulations based on synthetic and real data. The computer code implements the approach described above and was implemented in Matlab and C++. In all simulations, the nodes have been deployed in a square area. Their communication range was varied to obtain the desired value for the diameter of the network. Networks made up of several independent clusters were discarded.

\begin{figure}
  \centering
  \includegraphics[width=0.95\columnwidth]{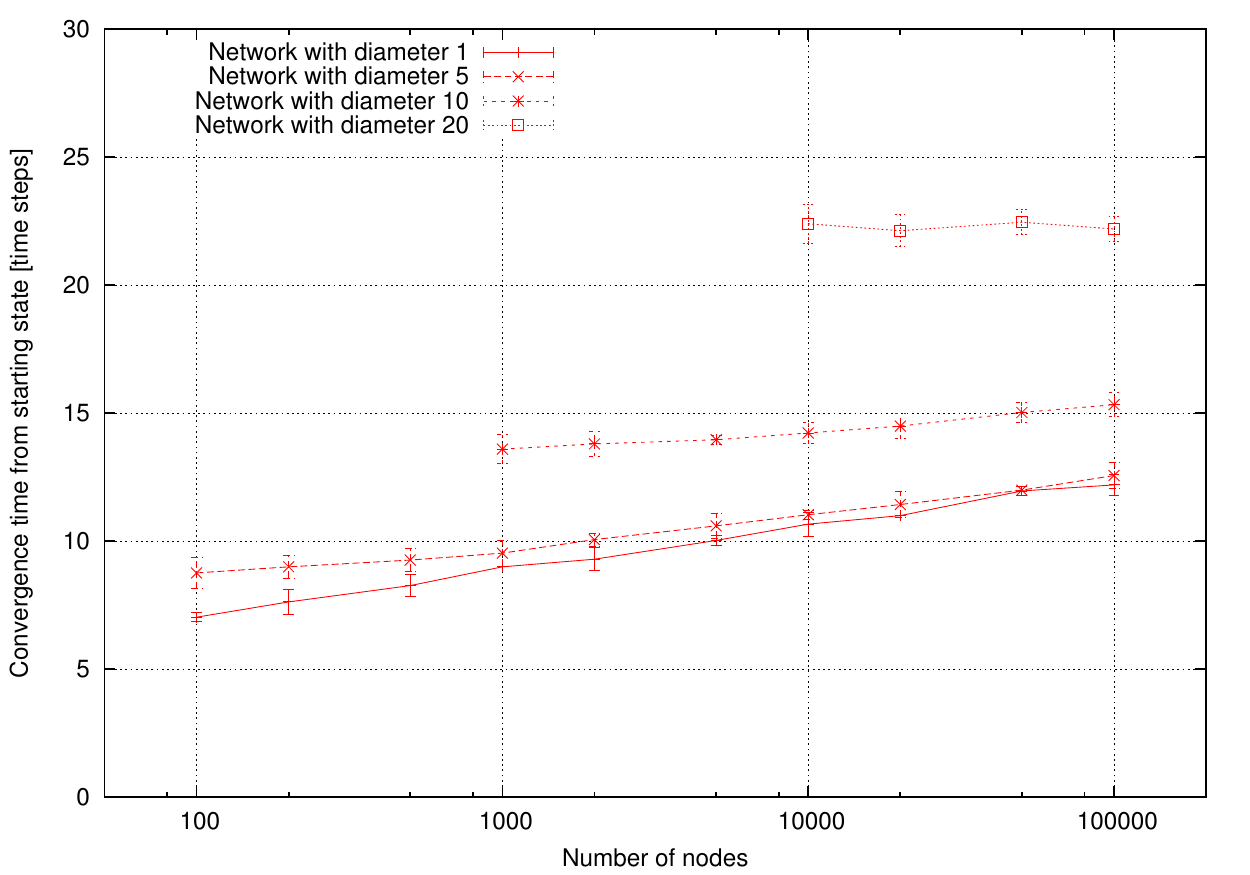}
  \caption{Convergence of network starting from a clean state (geometric random graph, nodes initialized with random values).}
  \label{fig:convsimple}
\end{figure}

\subsection{Data Generation}
\label{sec:experiments}

As far as the authors are aware there is no data available in the
public domain that describes both the structure and the change in load
over some time period for a power grid. To show the effectiveness of
our approach we have generated this data ourselves, below we explain
how this is done and we show the effectiveness of the proposed
distributed algorithm for calculating robustness of an operational
grid.

The computation of the system robustness of a power grid requires data
describing its topology (i.e., interconnection of nodes with lines),
the electrical properties of its components (i.e., admittance values of
the transmission lines), information about the nodes (i.e., number and
their types), and finally their generation and load values. The IEEE
power test systems~\cite{TestCaseRef} provide all of these data, the
IEEE 118 power system provides a realistic representation of a real
world power transmission grid consisting of 118 nodes and 141
transmission lines. We use this as a reference power grid.

The IEEE 118 power system gives information about the topology of the
power grid. The loading profile provided with the grid
topology~\cite{TestCaseRef} gives a representative load for the
network, but only for one moment in time. However, in practice, the
topology of a power grid remains generally unchanged over time (except
for the maintenance, failure and extension of the grid) while the
generation/loading profile varies over time. This changing nature of
the loading profile (and accordingly the generation profile) results
in a varying robustness of the system over time. Therefore simulating
the robustness profile of a power grid for a whole day requires a
demand profile belonging to the whole day.

To obtain a varying robustness for the IEEE 118 power system, we
randomly choose 10\% of the power generation nodes of the power system
which are then fed with synthetic (generated) demand
profiles. The demand values of other power generation nodes remain
unchanged. The demand profiles are generated based on an actual load
profile for a day of the Dutch grid on January 29, 2006. The demand at
the corresponding point in the Dutch grid is sampled per 15 minutes
during the whole day. Figure~\ref{fig:loadingProfileDutchGrid} shows the
demand profile. Based on this actual demand profile, additional
synthetic demand profiles are generated by (i) first introducing
random noise to the actual demand profile, and
(ii) then by smoothing the curve out with a moving
average~\cite{MovingAverages1962} with a window size of
10. Figure~\ref{fig:loadingProfileDutchGrid} illustrates the actual
demand profile and two other synthetically generated demand curves.

\begin{figure}
  \centering
  \includegraphics[width=0.95\columnwidth]{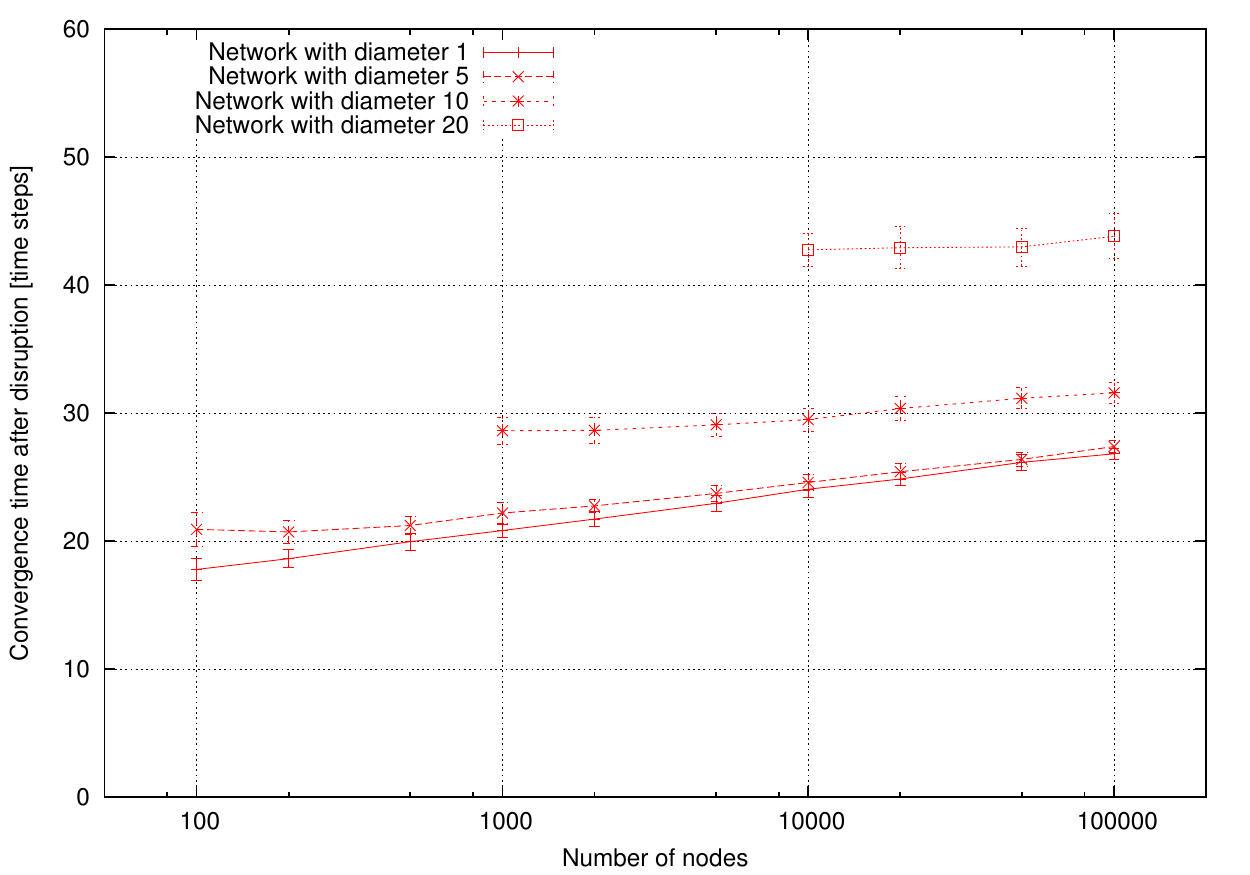}
  \caption{Convergence of network after a disruption (geometric random graph, half of the nodes change their values after initial network convergence).}
  \label{fig:convttl}
\end{figure}

\subsection{Influence of Communication Topology}

The underlying communication network for a smart grid can be implemented in a number of ways, mapping to different communication topologies. For example, one might choose to use the internet backbone, allowing any-to-any communication in the network, leading to a fully connected graph. In the first experiment, we have initialized the network with a set of random variables and recorded the time when the aggregated sum converges to the same value on all nodes. As seen in Figure~\ref{fig:convsimple}, fully connected networks lead to the fastest aggregate computation. In a second experiment, once the network stabilized, we introduced a change in the form of half of the nodes in the network changing their value to a different one. Again, we recorded the time until the network stabilized after this change. As expected, Figure~\ref{fig:convttl} shows that fully connected networks stabilize the fastest after a disruption. 

These results assume the internet backbone to work perfectly and able to route the high level of traffic generated. A more realistic scenario is considering that the various data collection points obtain data from the individual consumers via some radio technology (for example GPRS modems) and are themselves connected to the internet backbone. To keep the traffic in the network to a minimum, the data collection points only communicate with their network-wise first order neighbors, leading to a mesh network deployment type. As seen in Figure~\ref{fig:convsimple} and Figure~\ref{fig:convttl}, the diameter of the network clearly has the major impact factor on the results, confirming the theoretical convergence results. The information needs at least $O(D)$ time steps to propagate through the network. The constant in the $O()$ notation is influenced on one hand by the average connectivity in the network (a node can only contact a single neighbor per time step, slowing information dissemination) and the push-pull communication model on the other (a node may be contacted by several neighbors during a time step, speeding up information dissemination).

\begin{figure}
  \centering
  \includegraphics[width=0.95\columnwidth]{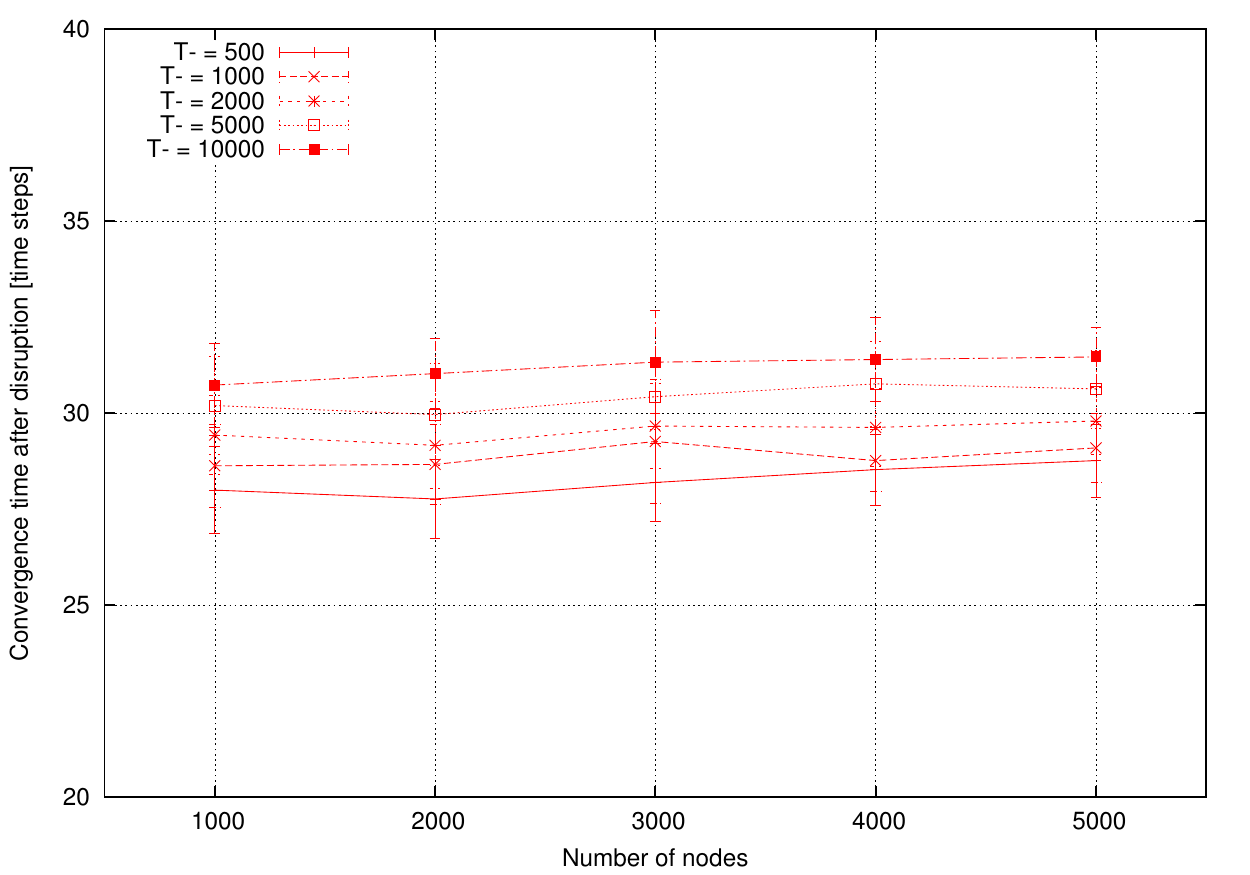}
  \caption{Influence of $T$ parameter (random geometric graph, 10-hop network, half of the nodes change their values randomly after initial network convergence).}
  \label{fig:ttl}
\end{figure}

\subsection{Scalability Aspects}
\label{subset:scalability}

One of the main characteristics of our approach is that the algorithm we propose scales very well with the number of nodes in the network. As seen again Figure~\ref{fig:convsimple} and Figure~\ref{fig:convttl}, the number of nodes has little influence in the final results (influencing only as $O(\log N)$). The simulation explored a space in which we varied the number of nodes over four orders of magnitude and the results hint that tighter boundaries might exist then the ones we proposed in this paper. We noticed that for a fully connected network, the recovery time varies with $34\%$ between a network with 1000 nodes and one with 100000 nodes, while the variation drops to a mere $2.4\%$ for a 20-hop network varying from 1000 nodes to 100000 nodes.

These results are very important for the smart grid application type. As the network will be linked to a physical space (a country or in general, a region), fully covering it, the diameter of the network is expected to, at most, decrease with the addition of new nodes. Intuitively, when thinking of nodes as devices with a fixed transmission range, adding more devices in the same region may lead to shorter paths between various points. The aggregate computation approach we propose shows on one hand an almost invariance to the increase in the number of nodes in the network and a linear variation with the diameter. These properties are essential for any solution that needs to take into account that the number of participants in the grid will most likely increase over time.

We are also interested in understanding the effects the time-to-live of the negative fields has on the convergence and scalability properties. We have considered a 10-hop network with 1000 to 5000 nodes and varied the time-to-live for negative values between 500 and 10000. Figure~\ref{fig:ttl} confirms Lemma~\ref{lemma:removal} with respect to the $\log T$ term. As the data shows, the convergence time was affected very little by the chosen parameters. As expected, the diameter of the network has the larger influence in this mechanism.

\subsection{Robustness Metric Computation}
\label{subset:robust-metr-comp}

\begin{figure}
  \centering
  \includegraphics[width=0.95\columnwidth]{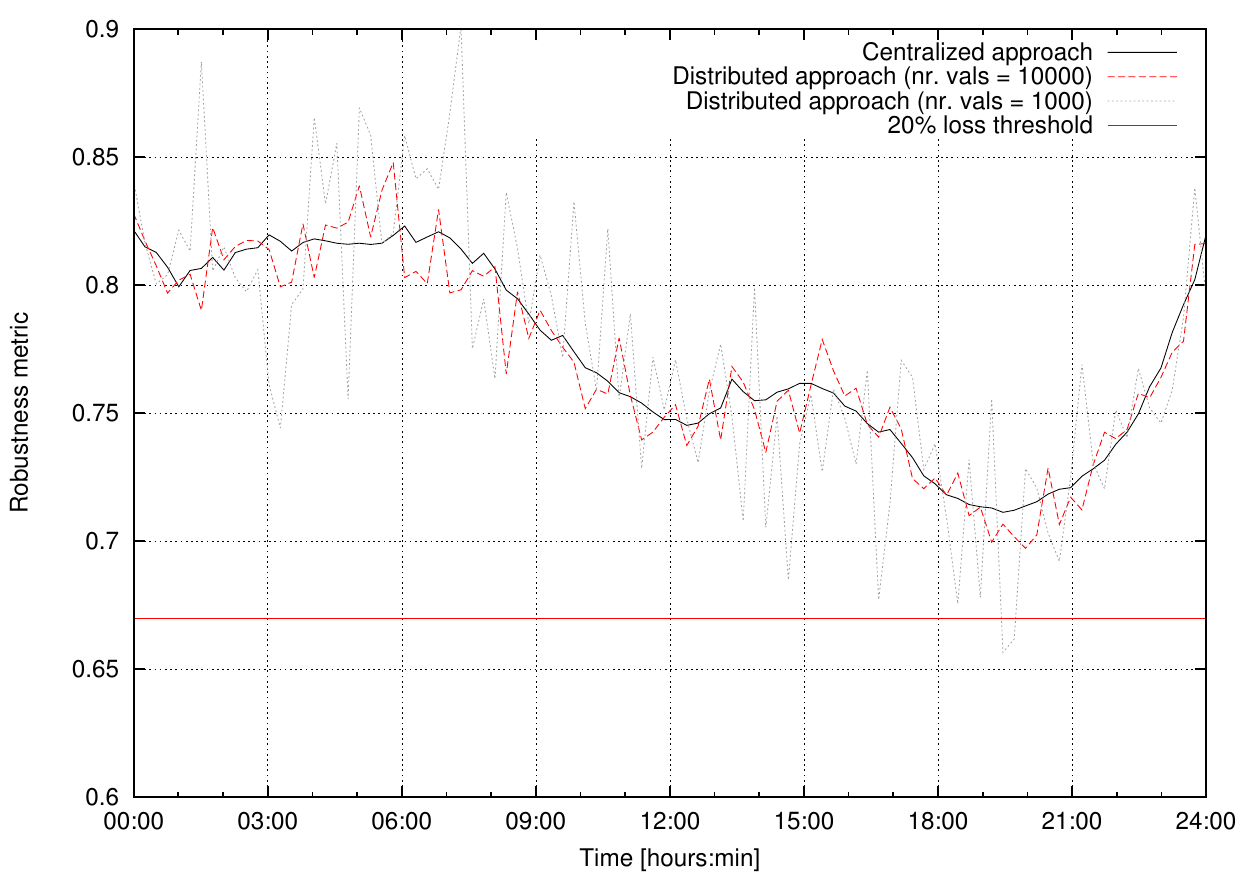}
  \caption{Robustness metric (centrally computed values versus two runs of the distributed algorithm, synthetic data - see Section~\ref{sec:experiments}, each point represents network data after convergence; the line $R_{CF}=0.67$ illustrates a critical threshold below which line failures can lead to large blackouts in this particular network).}
  \label{fig:robust}
\end{figure}

Figure~\ref{fig:robust} shows the distributed computation method performing with real data sets, obtained through the method described in Section~\ref{sec:experiments}. We plotted the results of two simulation runs versus the ground truth data, obtained via centralized computation. The length of the value vector was varied from 1000 values to 10000 values, the results confirming that precision can be set to the desired value, independent of the network topology and size. 
When using a vector of 1000 elements, we obtained a mean relative error of $3\%$ (maximum relative error $11\%$ with a standard deviation of $2.6\%$). Using a larger vector (10000 elements) we were able to obtain a mean relative error of $1\%$ (maximum relative error of $4\%$ with a standard deviation of $0.8\%$). These figures are very good, taking into account that they result from a combination of distributed computations with all the fault tolerant mechanisms enabled.

The figure also includes a line (with robustness value 0.67) that
illustrates the critical threshold, set by the grid operator. If the
robustness metric drops below this value then a power line failure can
lead to a blackout that effects more then 20\% of the power grid. This
threshold value was obtained by running cascading failure simulations
on the IEEE 118 power grid system using targeted attacks (i.e., we
considered a worst case scenario). We refer to~\cite{Koc:571} for a
structured methodology for determining such thresholds. 

The critical threshold chosen above, that effects more 20\% of the
power grid, is more or less arbitrary and mainly chosen for
illustration purposes. In practice various other factors have to be
taken into account by grid operators (line capacities, maintenance
cycles) to determine realistic threshold values, but this should
illustrate the feasibility of the approach as it clearly shows that
the error rate of the distributed algorithm is much smaller then
minimal required drop in robustness value that is needed to meet the
threshold.

Besides the quantitative values shown in Figure~\ref{fig:robust} we
would like to point that our approach is different from traditional
approaches that try to capture the global state of the network and
then take decisions centrally (see
Section~\ref{sec:syst-health-monit}). Our approach pushes the
computation of the robustness metric \emph{in the network}, its
results being available at each node as soon as the computations
converge. This mechanism can be easily used as a measurement phase,
leading to the possibility of implementing distributed control loops
on top of it.

\section{Monitoring and Cascading Failures in Power Grids}
\label{sec:syst-health-monit}

Three types of related work on monitoring the state of a power grid
can be distinguished: (i) metrics that aim to quantify the
vulnerability of the power grid against cascading failures, (ii)
simulation models that aim to predict the impact of node/line outages and
(iii) sensor networks that aim to capture the operative state of the
power grid. 

There exists a significant body of work on defining metrics that
assess the vulnerability of the power grid against cascading
failures. Most studies deploy a purely topological or an extended
topological approach mainly relying on graph theoretical measures such
as betweenness centrality~\cite{Mieghem2006}. However, these
studies~\cite{Kim2007, Chen2006, Bompard2009, Bompard10} only focus on
the topological properties of power grids and fail to take the
operative state of the network into account. In effect this means that
such metrics cannot be used to assess the change in vulnerability of
operational power grids. In addition to these topological approaches,
others~\cite{Youssef2011, Bao2009} propose measures relying on
simulation models. Although, these metrics incorporate also the
operative state of a power network, it is very challenging to deploy
them to quantify the system's resilience against cascading failures in
(near) real-time because their computation requires full knowledge of
the power grid state in order to simulate cascades. Our earlier
work~\cite{Koc:862,Koc:373} (also see
Section~\ref{sec:robustn-metr-monit}) forms a noticeable exception to
this, since it defines a metric that considers both the topological
and the operative state of a power grid, while not requiring any
computationally expensive tasks (e.g., computing the full network
state in order to simulate cascades in the network).

Grid operators traditionally assess the network operation by relying
on flow based simulation models (i.e., N-x contingency
analysis~\cite{Baldick2009}). These models take the operational
behavior of the power grid into account. Grid operators can calibrate
the model to match the power grid of interest and run various
scenarios to assess the impact of one or two lines failing. There are
two problems with such tools: they depend on the knowledge of the grid
operator who determines which failure scenarios to explore. In
addition, due to the complexity of the simulation models it is
typically not possible to run scenarios where more than two components
fail. The monitoring approach proposed in this paper may complement
current grid operator practices.

There are numerous papers that describe distributed architectures that
can be used to monitor the state of the power grid. However, these
typically focus on the issue of data
collection~\cite{zanikolas2005taxonomy,yang2006power,gungor2010opportunities,dan2004wide,bakken2007towards,moslehi2010reliability,zhang2012monitoring,grilo2014load}
(i.e., loading levels of power lines, phase angles etc.) and do not
use any meaningful data aggregation mechanisms to quantify the
resilience with respect to cascading failures of the whole power grid.
In conclusion, as far as the authors are aware, there are no power
grid monitoring approaches that assess the vulnerability, with respect
to cascading failures, of an operational power grid in near
real-time.

\section{Conclusions}
\label{sec:conclusions}

In this paper we introduced a novel distributed computation framework
for network aggregates and showed how it can be used to assess the
resilience with respect to cascading failures of an operational power
grid in near real-time. We have enhanced a class of fast gossiping
algorithms~\cite{boyd2005gossip} with self-stabilizing mechanisms to
counter run-time network dynamics. To showcase the capabilities of our
approach, we exemplified how the robustness metric introduced
in~\cite{Koc:373,Koc:862} can be computed fast and reliable in a
distributed network - IEEE 118 power grid.

Our contribution has a number of desirable properties such as
scalability and robustness. Simulation results performed with both
real and synthetic data show that our approach achieves very fast
convergence times, influenced mainly by the diameter of the network
and only logarithmically by the number of nodes in the network. This
property is very important in the context of smart grids, where the
number of nodes deployed over a given area (a region or a country) is
expected to increase in the next few decades.

The precision of the computations can be fixed by modifying the size
of the messages exchanged in the network. This is a crucial property
for scalability, as the size of the messages is \emph{not} a function
of the number of nodes in the network. More importantly, the
computation error scales as $O(1/poly(N))$, meaning that the more
nodes a network has, the smaller the final error is. Finally, our
scheme preserves the anonymity of the participants in the network, as
it does not rely on unique identifiers for the nodes of the network.

The main message of this paper can be summarized in that we showed
that it is possible to compute complex aggregates of the operational
state of the nodes in a network in a fully distributed manner, fast and reliable at
runtime. As automatic control systems always include a measurement
phase, we see our contribution as the perfect candidate for the
measurement block for an automated distributed control scheme. While
this paper focused on the measurements of network properties, future
work will investigate the actuation part triggered by the availability
of results given by different power grid metrics.

\section*{Acknowledgment}

{\small
This work was partly funded by the NWO project \emph{RobuSmart: Increasing
  the Robustness of Smart Grids through distributed energy
  generation: a complex network approach}, grant number
647.000.001 and by the Rijksdienst voor Ondernemend Nederland 
grant TKISG01002 \emph{SG-BEMS}.


\end{document}